\documentclass[journal,11pt]{IEEEtran}

\usepackage{bm}
\usepackage{amsmath}
\usepackage{amsthm,letltxmacro,changepage}
\usepackage{amssymb}
\usepackage{algorithm} 
\usepackage{algpseudocode} 
\usepackage{mathtools}

\DeclarePairedDelimiter\floor{\lfloor}{\rfloor}
\usepackage{rotating}
\usepackage{tabularx}
\usepackage[square,sort,comma,numbers]{natbib}
\usepackage{url}
\usepackage{multicol}
\newtheorem{theorem}{Theorem}[section]

\begin{document}


\title{
CompactChain:An Efficient Stateless Chain for UTXO-model Blockchain
}

\author{B Swaroopa Reddy and T Uday Kiran Reddy
\thanks{
Both the authors are with Department of Electrical Engineering, Indian Institute of Technology Hyderabad, Hyderabad, Telangana 502285, India
}
}
\markboth{This work is accepted at the journal of Springer 
Frontiers of Computer Science.}%
{}

\date{\today}

\maketitle

\begin{abstract}
In this work, we propose a stateless blockchain called CompactChain, which compacts the entire \textit{state} of the UTXO (Unspent Transaction Output) based blockchain systems into two RSA accumulators. The first accumulator is called Transaction Output (TXO) commitment which represents the \textit{TXO set}. The second one is called Spent Transaction Output (STXO) commitment which represents the \textit{STXO set}. In this work, we discuss three algorithms - $(i)$ To update the TXO and STXO commitments by the miner. The miner also provides the proofs for the correctness of the updated commitments;
$(ii)$ To prove the transaction's validity by providing a membership witness in TXO commitment and non-membership witness against STXO commitment for a coin being spent by a user;
$(iii)$ To update the witness for the coin that is not yet spent; 
The experimental results evaluate the performance of the CompactChain in terms of time taken by a miner to update the commitments and time taken by a validator to verify the commitments and validate the transactions. We compare the performance of CompactChain with the existing state-of-art works on stateless blockchains.  CompactChain shows a reduction in commitments update complexity and transaction witness size which inturn reduces the mempool size and propagation latency without compromising the system throughput (Transactions per second (TPS)).\\ 
\end{abstract}
\section{Introduction}
Blockchain has emerged as a  decentralized and trustless technology for both cryptocurrencies and smart contract applications. An anonymous person named  Satoshi Nakamoto introduced Bitcoin \cite{bitcoin} as a Peer-to-Peer  ($P2P$) network that records an append-only immutable ledger using cryptographic digital signatures  and hash functions. Blockchain-based smart contract platforms  like Ethereum \cite{ethereum}, Hyper ledger fabric \cite{hyperledger} open the applications of blockchain in many areas like Internet-of-Things (IoT) \cite{IoT}, supply chain management \cite{supplychain}, healthcare \cite{medrec}, agriculture  \cite{agriculture}, energy trading \cite{energy}, etc. 

A transaction is a fundamental entity in the blockchain ledger. It is defined as transferring the ownership of a coin from one party
to the other through digital signatures \cite{bitcoin}. In a transaction, the destinations of ownership transfer are called \textit{outputs}, and the sources of ownership are called \textit{inputs}. A subset of previous \textit{outputs} is spent as \textit{inputs} of the transaction. The transactions contain multiple \textit{inputs} and \textit{outputs}. The set of the Unspent Transaction Outputs
(UTXO) is called the \textit{state} of a blockchain.
The \textit{state} changes dynamically as a subset of previous \textit{outputs}  (or \textit{inputs}) are being spent, and create new \textit{outputs}.

The blockchain network is jointly maintained by two different types of nodes - miners, and users. The \textit{miners} are the special nodes who invest their computational power to solve the Proof-of-Work (PoW) \cite{bitcoin} puzzle to create immutable blocks.
The \textit{full nodes} \cite{bitcoin_full_node} and \textit{miners} in the network act as \textit{validators}. The validators store the complete UTXO set (state) \cite{UTXO} to verify the validity of blocks and transactions.  So, bitcoin is an example of a UTXO based stateful blockchain. The full nodes store the UTXO set in the chainstate directory, a level dB database of the bitcoin core \cite{bitcoin_github}. The validators also store the complete blocks from the genesis to the present block and help new nodes joining the network to synchronize with the blockchain's current state.

The blockchain \cite{bitcoin_charts} and UTXO set \cite{bitcoin_charts} sizes are ever-growing and increase the storage burden on validators and reduce transaction verification performance. Consequently, some users run a Simplified Payment Verification (SPV) node or light client \cite{bitcoin}. The SPV nodes  rely on the full nodes to verify the transaction by querying the Merkle path from the transaction to the Merkle root listed in the block header. Further increase in the size of the UTXO set might exceed the RAM capacity of the validator leads to an increase in the count of the SPV nodes compared to the validator nodes, which underpins the decentralization property of the blockchain system. 

Accumulator is a cryptographic primitive that produces a short binding commitment to a set of elements together with short membership/non-membership proofs for any element in the set. The membership witness is a single element, which allows us to prove the fact that a given element is included in the accumulator. Similarly, the non-membership witness allows us to prove the fact that a given element is not a part of the accumulator. The RSA accumulator \cite{accumulator, dynamic_acc, universal_acc} is an accumulator based on strong RSA assumption. 

In \cite{batching}, the authors propose a stateless blockchain for UTXO-model and account-model blockchains  using the batching techniques on trapdoor-less RSA accumulator \cite{dynamic_acc}, \cite{universal_acc}, \cite{Tremel}. In the UTXO-model framework, the block contains a commitment to the latest UTXO set called the \textit{accumulator state} or \textit{UTXO commitment}. The accumulator state is verified by the validators using Non-Interactive Proof of Exponentiation (NI-PoE). The users provide the membership witness for their coins being spent in transaction payloads. The miners and fullnodes  verify the transactions by checking the membership witnesses against the latest accumulator state. 
This framework uses batching techniques on RSA accumulators, namely \textit{batchDel} and \textit{batchAdd} to update the  commitment for every new block creation. The \textit{batchDel} operation deletes the spent coins or inputs of the transaction from the UTXO commitment. The \textit{batchAdd} operation adds the new coins or outputs to the UTXO commitment. However, the \textit{batchDel} operation aggregates membership witnesses for all inputs of the block into a single membership witness using \textit{shamirTrick}. The \textit{batchDel} is a sequential iterative operation with a complexity of  $\mathcal{O}(m^2)$. Where $m$ is the number of transactions (assuming single input per transaction). Moreover, it is not possible to exploit parallelism in the computation.  Since the complexity of \textit{batchDel} operation is quadratic in time, the efficiency drastically reduces with an increase in the total  number of inputs in a block. It severely affects the transaction throughput (TPS). Also, users need to rely on the service providers to update their membership witnesses with every block of transactions.

Minichain protocol  proposed in \cite{minichain} is  a light-weight stateless blockchain.  Minichain consisting of two commitments - STXO (Spent Transaction Output) commitment, and TXO (Transaction Output) commitment instead of a single UTXO commitment to avoid more complex \textit{batchDel} operation. The STXO commitment is an RSA accumulator to all the spent coins or inputs (called STXO set). The TXO commitment is based on the Merkle Mountain Range (MMR) \cite{petertodd} to outputs (called TXO set). The validators in Minichain also use the NI-PoE to verify the STXO commitment. While spending a coin, the user provides a non-membership witness (Unspent proof) against STXO commitment and a membership witness (Existence proof) in TXO commitment.  The existence proof consists of two Merkle proofs. First, the \textit{coin inclusion proof} to prove the coin is committed to the TMR (Transaction Merkle Root) of the block where the coin is generated. Second, the \textit{TMR inclusion proof} to prove that the TMR is committed to the latest TXO commitment. However, the complexity of the coin inclusion proof and TMR inclusion proof sizes are $\mathcal{O}(log_2 (m))$ and $\mathcal{O}(log_2 (L))$. Where $L$ is the length of the blockchain from the genesis block.  $L$ is ever-growing and increases network communication latency due to the large membership proof size, which further affects the TPS. This impact is huge in the blockchain networks with high block creation rate.

In this work, we propose a compact stateless blockchain for UTXO-model blockchain by compressing the entire \textit{state} of the blockchain into two RSA accumulators, one each for TXO commitment and STXO commitment. The miner updates the commitments by using two \textit{batchAdd} operations on inputs and outputs. The validators verify both commitments by checking the  NI-PoE proofs. The users provide transaction proof consisting of membership witness for a coin in TXO commitment and a non-membership witness against the STXO commitment. 
The membership witness is also an accumulator excluding the particular coin whose membership is to be proved. Consequently, membership proof is of constant size. The user needs to update the membership and non-membership proofs for their coins with every new block generation to make the proofs compatible with the latest commitments. The validator can efficiently verify the transaction proofs through the RSA group operations.

The main contributions of this work are the following - 
\begin{itemize}
\item We propose an RSA accumulator for the TXO commitment for a constant sized transaction membership proof in contrast to the ever-growing existence proof size in minichain protocol.

\item We implement CompactChain and compare the performance with Boneh's\footnote{We use the term Boneh's protocol for the stateless blockchain proposed in \cite{batching}.} and Minichain protocols. Comparing to Boneh, CompactChain has improved the efficiency of commitment update from $\mathcal{O}(m^2)$ to $\mathcal{O}(m)$. Comparing to Minichain, the transaction proof size has improved from $\mathcal{O}(log_2 (m)) + \mathcal{O}(log_2 (L)) + \mathcal{O}(1)$ to $\mathcal{O}(1)$.

\item Through simulation results, we show the performance improvement in network communication latency and TPS compared to Minichain due to reduced transaction proof size.
\end{itemize}

The rest of the paper is organized as follows - We discuss the related work in section 2. Section 2 provides the  preliminaries. In section 4, we present the system architecture. In section 5, we discuss the design of CompactChain protocol. Section 6 demonstrates the performance evaluation of CompactChain in comparison with Boneh's and Minichain protocols. In section 7, we conclude our work and discuss the future directions of research. 
\section{Related Work}
The distributed coding theory techniques like erasure-codes \cite{low-storage} and fountain codes \cite{SeF} have been proposed for storage efficiency of the blockchain node
by reducing the storage cost and still contribute to bootstrap a new node joining the network. In Dynamic distributed storage system \cite{dynamic}, the blockchain nodes are allocated into dynamic zones and the nodes in each zone store a share of private keys using Shamir's secret-sharing \cite{shamir} for encrypting the block data and apply a distributed storage codes such as \cite{exact-regeneration}, \cite{optimal-LRC} for reducing the storage cost. 

A snapshot-based block pruning technique has been proposed in \cite{coinprune} to prune archived blocks by creating a \textit{snapshot} of the state at regular intervals. A snapshot-based consensus protocol for bitcoin-like blockchains has been proposed in \cite{Rollerchain}, where the miners create a block by providing a non-interactive proofs of storing a subset of the past state snapshots. However, these techniques solve the problem of storing the  historical blocks.

A stateless client concept proposed in \cite{vitalik} for Ethereum blockchain, where full nodes only store a state root and miners broadcast witness (a set of the Merkle branches proving the data values in block) along with the block. The validators download and verify these expanded blocks. In \cite{petertodd}, the author proposes  low-latency delayed TXO commitments based on Merkle Mountain Range (MMR) for committing to the state of all transaction outputs. In order to append a new output requires  fewer storage requirements ($log_2 (n)$ mountain tips). The added output could not be removed from the MMR, instead updates the status of the spent  output. While spending a coin, each transaction would be accompanied by a Merkle proof consisting of a Merkle path to the tip of a tree such that the outputs being spent were still unspent. 

EDRAX \cite{EDRAX} proposes a stateless transaction validation for UTXO based blockchain using the Sparse Merkel Tree (SMT) \cite{SMT}. The coin which is being spent also includes a witness of unspent. However, the witness for a transaction being spent depends on the Merkle proof of size $log_2 (m)$.

In \cite{batching}, the authors proposed a stateless blockchain scheme based on the batching and aggregating techniques for the accumulators of unknown group order. In this scheme, each block contains an accumulator, which represents a commitment to the current UTXO set. This commitment is constructed by leveraging the batching techniques for membership and non-membership proofs for the set of transactions included in the block. However, the complexity of this commitment update is of  $\mathcal{O}(m^2)$.  Where $m$ is the number of transactions (assuming single input per transaction). The efficincy in computing the commitment reduces with the increase in number of average transactions per block.
\section{Preliminaries}\label{sec3}
\subsection{State of an UTXO based Blockchain}
In a UTXO-model blockchain design, the \textit{state} or \textit{UTXO set} \cite{UTXO} is a collection of the Transaction Outputs (TXO) that are unspent at a particular moment of time. Whenever a new transaction is created, it consumes a subset of the current UTXO set through \textit{inputs} and creates new UTXOs through \textit{outputs}. A validator in a stateful blockchain needs to keep a copy of the \textit{state} as a \textit{chainstate} database to verify the transaction's validity. The validators are required to update the \textit{state} with a block of transactions by deleting the UTXOs associated with inputs from the \textit{UTXO set} and adding new outputs to the \textit{UTXO set}.

Let $S_n$ be the \textit{state} of a blockchain until block height $n$ and $\bullet$ represents a state transition function. When a new block $B_{n+1}$ is created with a set of transactions included in it, then the \textit{state} of the blockchain is
\begin{align}
S_{n+1} &= S_n \bullet B_{n+1} \nonumber \\
&=(\dots ((S_{0} \bullet B_{1}) \bullet B_{2}) \bullet \dots  \bullet B_{n+1})
\end{align} 
\subsection{Stateless Blockchain}
In contrast to the stateful blockchain, in stateless blockchain \cite{batching, minichain} a cryptographic commitment like RSA accumulator \cite{batching} to the \textit{UTXO set} is stored in every block header. The validators no longer need to store the complete \textit{UTXO set} in their storage to verify the transactions, instead, the witness provided by each user for their transactions are used to verify the transaction's validity against the latest commitment to the \textit{state}. Any miner update the commitment from inputs and outputs of the transactions in the block and output a proof of correctness. The miner propagate the proof to other nodes in the network.

Let $C_n$ be a commitment to the \textit{UTXO set} in block $B_n$ and $\diamond$ be the commitment update function, then the new commitment for the block at height $n+1$ is obtained as
\begin{equation}
C_{n+1} = C_{n} \diamond B_{n+1}
\end{equation}

The steps to run the network in a stateless blockchain are as follows - 
\begin{enumerate}
\item Users broadcast transactions along with corresponding witnesses to all the nodes.
\item Miners collects new transactions into the block and update the commitment concerning the commitment of the previous block. It also outputs a proof for correctness of the commitment.
\item Each miner works on finding the Proof-of-Work (PoW) for its block.
\item When a miner finds a PoW, it broadcasts the block along with the witnesses to the transactions in the block to other nodes in the network.
\item Validators accept the block only if all the transactions in it are valid based on the transaction witnesses and the proof of correctness to the commitment is verified. 
\item Nodes express their acceptance of the block by working on creating the next block in the
chain.
\item The owners of the unspent coins update their transaction proofs based on the newly accepted block of transactions in the chain.
\end{enumerate}

\subsection{RSA Accumulator}
\subsubsection{Cryptographic Assumptions}
An RSA accumulator \cite{accumulator}, \cite{Tremel} is a cryptographic primitive that generates a short commitment to a set of elements with efficient membership and non-membership proofs for any element of the set. 

The RSA accumulator requires the generation of a  group of unknown order in which strong RSA and root assumption \cite{batching} holds, and it is based on the modular exponentiation with an RSA modulus. 

Let $\mathbb{G}$ be a group of unknown order, $g \in \mathbb{G}$ be a generator of the group and $N$ be an RSA modulus such that $N=pq$, where $p$ and $q$ are strong primes \cite{strongprimes}. 

Let \textit{GGen} is an algorithm that generates the above public parameters and $\mu(\lambda)$ is a negligible function in the security parameter $\lambda$.\\
\textbf{ Strong RSA Assumption} \cite{accumulator, dynamic_acc, batching}: Given RSA modulus $N$ (of size $\lambda$)) and a random $y \in \mathbb{Z}_N$, then it is computationally hard to find $x \in \mathbb{Z}_N$ and $l > 1$ such that $x^l \equiv y$ mod $N$. i.e., for all probabilistic polynomial time (PPT) adversary $\mathcal{A}$
\begin{align*}
Pr
\begin{bmatrix}  & \mathbb{G} \leftarrow GGen(1^{\lambda}), g \in \mathbb{G};\\ x^l = y : & \\ & (x, l) \in \mathbb{G} \times \mathbb{Z}_N \leftarrow \mathcal{A}(\mathbb{G},g)
\end{bmatrix}
\leq \mu(\lambda) 
\end{align*}
The elements to be accumulated must be prime numbers in order to be  collision-free \cite{accumulator} under strong RSA assumption.  Let $S = \{e_1, e_2, \dots, e_m \}$ be a set of elements to be accumulated, then the accumulator of elements of $S$ is
\begin{equation}
A = g^{x_1 . x_2. \dots . x_m}
\end{equation}
where,
\begin{equation}
x_i = H_{prime} (e_i)
\end{equation}
Where, $H_{prime}$ a random hash to prime function \cite{batching}.

The dynamic accumulator \cite{dynamic_acc} supports the addition or deletion of the elements to the accumulator. 

Let $S_{new} =  \{e_{m+1}, e_{m+2}, \dots e_n \}$ be a set of new elements to be added, then the batch addition of set $S_{new}$ update the accumulator to $A_{new} \leftarrow A^{x^{*}}$, where $x^{*} = x_{m+1} . x_{m+2}. \dots . x_n$.
\subsubsection{Membership and Non-membership witness}
The function  membership witness $w_i$ \cite{universal_acc},{batching} of an element $e_i \in S$ is simply an accumulator without the element $e_i$
\begin{equation}
w_i = g^{\prod_{j = 1, j \neq i}^m x_j} 
\end{equation}
The membership of $e_i$ in $A$ is verified by checking
\begin{equation}
(w_i)^{x_i} \stackrel{?}{=} A  
\label{mem}
\end{equation}
Let $e \notin S$ and $x = H_{prime}(e)$, then non-membership witness $u_x$ of $e$ in $S$ uses the fact that $gcd\left(x, x^{*} = \prod_{i = 1}^{m} x_i \right) = 1$. The Bezout coefficients  $a$ and $b$ such that $ax + bx^{*} = 1$ gives the non-membership witness for $e$ is
\begin{equation}
u_x = (d, b) = (g^a, b)
\end{equation}
The non-membership witness is verified by checking
\begin{equation}
    d^x A^b \stackrel{?}{=} g
    \label{nonmem}
\end{equation}
\subsubsection{Accumulator Security (Undeniability)}
\label{sec}
The universal accumulator \cite{universal_acc} is a dynamic accumulator with efficient update of the membership and non-membership proofs. We use the trapdoor less universal RSA accumulator \cite{batching}, \cite{secure_acc}, as there is no single trusted manager exists in distributed blockchain system and updates are processed in batches.

Let generator $g \in \mathbb{G}$ and $\mu(\lambda)$ is a negligible function. The functions \textbf{VerMemWit} and \textbf{VerNonMemWit} verify the membership and non-membership witnesses as per the equations \eqref{mem} and \eqref{nonmem} respectively.
A dynamic universal accumulator is secure \cite{universal_acc}, \cite{secure_acc}, \cite{batching} if, for all probabilistic polynomial time (PPT) adversary $\mathcal{A}$,
\begin{align*}
Pr
\begin{bmatrix} \mathbb{G} \leftarrow GGen(1^{\lambda}); g \in \mathbb{G}; (A,x,w_x,u_x)\leftarrow \mathcal{A}(pp, g); \\ \textbf{VerMemWit}(x, w_x, A) \wedge \textbf{VerNonMemWit}(x, u_x,A) 
\end{bmatrix} \\
= \mu(\lambda)
\label{eq:undeniability} 
\end{align*}
In other words, it is computationally infeasible to find both membership and non-membership proofs for an element $x \in S$. This statement is equivalent to it is computationally infeasible to find a non-membership witness for $x \in S$ or a membership witness for $x \not \in S$.
\subsubsection{Non-interactive proof of exponentiation (NI-PoE)}
Let $u, w \in \mathbb{G}$, the proof of exponentiation \cite{batching}, \cite{vdf} in the Group $\mathbb{G}$, when both the prover and verifier are given $(u,w, x \in \mathbb{Z})$ as inputs and prover wants to convince the verifier that $w=u^x$. Given base $u$, exponent $x$, and modulus $N$, the complexity of computing $u^x$ mod $N$ is $\mathcal{O}(log(x))$ \cite{mod-exp}. In PoE protocol, the verifier's work is much less than computing $u^x$ mod $N$, when $x \in \mathbb{Z}$ is much larger than 
$|\mathbb{G}|$. The PoE can be made non-interactive using the Fiat-Shamir heuristic \cite{Fiat-Shamir}.\\

\boxed{
\parbox[t]{3 in}{ \textbf{NI-PoE}  \cite{batching} \\
$.............................................................................$ \\
Public known $\{x, u, w:u^x = w \}$\\
$\#$ Prove NI-PoE
\begin{algorithmic}[1]  
\Procedure{\textit{ProveNI-PoE}}{$x$, $u$, $w$}
\State $l \gets H_{prime}(x, u, w)$ 
\State $q \gets \floor{\frac{x}{l}}$
\State return $Q \gets u^q$ 
\EndProcedure
\end{algorithmic}
\vspace{0.35cm}
$\#$ Verify NI-PoE
\begin{algorithmic}[1]
\Procedure{\textit{VerifyNI-PoE}}{$x$, $u$, $w$, $Q$}
\State $l \gets H_{prime}(x, u, w)$
\State $r \gets$ $x$ mod $l$
\State $Q^l u^r \stackrel{?}{=} w$
\EndProcedure
\end{algorithmic}
}}
\section{System Architecture}
\begin{table}[!h]
\caption{Parameters used in CompactChain}
\begin{center}
\begin{tabular}{ l l}
\hline
\textbf{Symbols}& \textbf{Description} \\
\hline
$B_n$ & Block at height $n$ \\
$T_n$ & Set of all transactions in $B_n$ \\
\textit{TXO set} & Set of all transaction outputs \\
\textit{STXO set} & Set of all spent transaction outputs \\
$TXO_n$ & Transaction outputs (TXOs) in $B_n$ \\
$STXO_n$  & Spent Transaction outputs (STXOs) in $B_n$ \\
$TXO_{k:n}$ & All TXOs from block $B_k$ to block $B_n$ \\
$STXO_{k:n}$ & All STXOs from block $B_k$ to block $B_n$ \\
$TXO\_C_n$ & Commitment to \textit{TXO set} in block $B_n$ \\
$STXO\_C_n$ & Commitment to \textit{STXO set} in block $B_n$ \\
$\Pi_{TXO}$ & NI-PoE proof for $TXO\_C_n$ \\
$\Pi_{STXO}$ & NI-PoE proof for $STXO\_C_n$ \\
$w_n(x)$ & Membership witness to prove $x \in$  \textit{TXO set} \\
$u_n(x)$ & Non-membership witness to prove  \\
& $x \not \in$ \textit{STXO set} \\
\hline
\end{tabular}
\label{table:symbols}
\end{center}
\end{table}
\begin{figure}[b]
  \includegraphics[width=1\columnwidth, height=0.6\columnwidth]{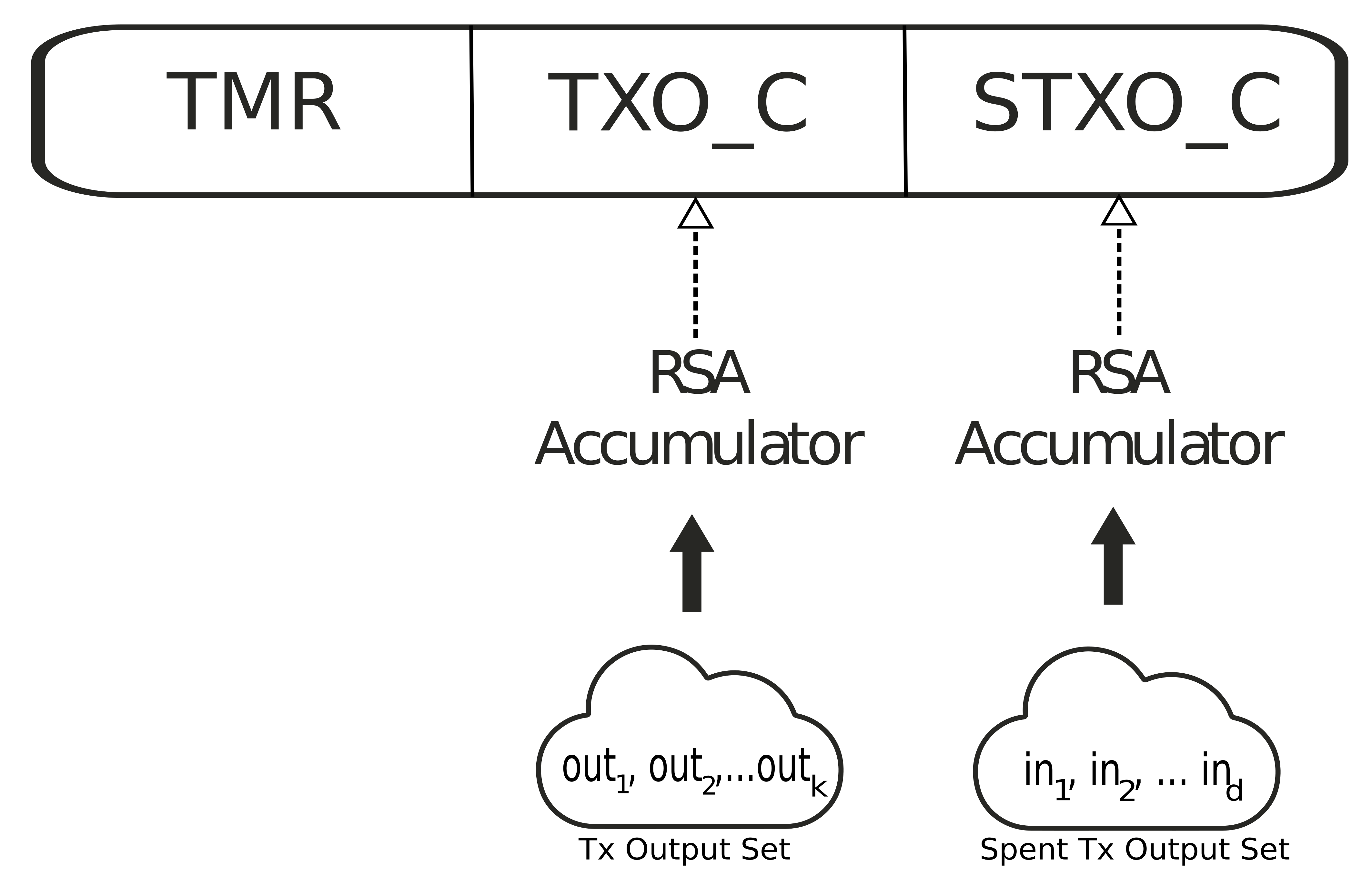}
   \caption{Block header composition mainly contains two RSA accumulators - $TXO\_C$ (TXO commitment) and $STXO\_C$ (STXO commitment)}
\label{fig:arch_a}
\end{figure}
\begin{figure*}[t]
  \includegraphics[width=1\linewidth, height=0.6\linewidth]{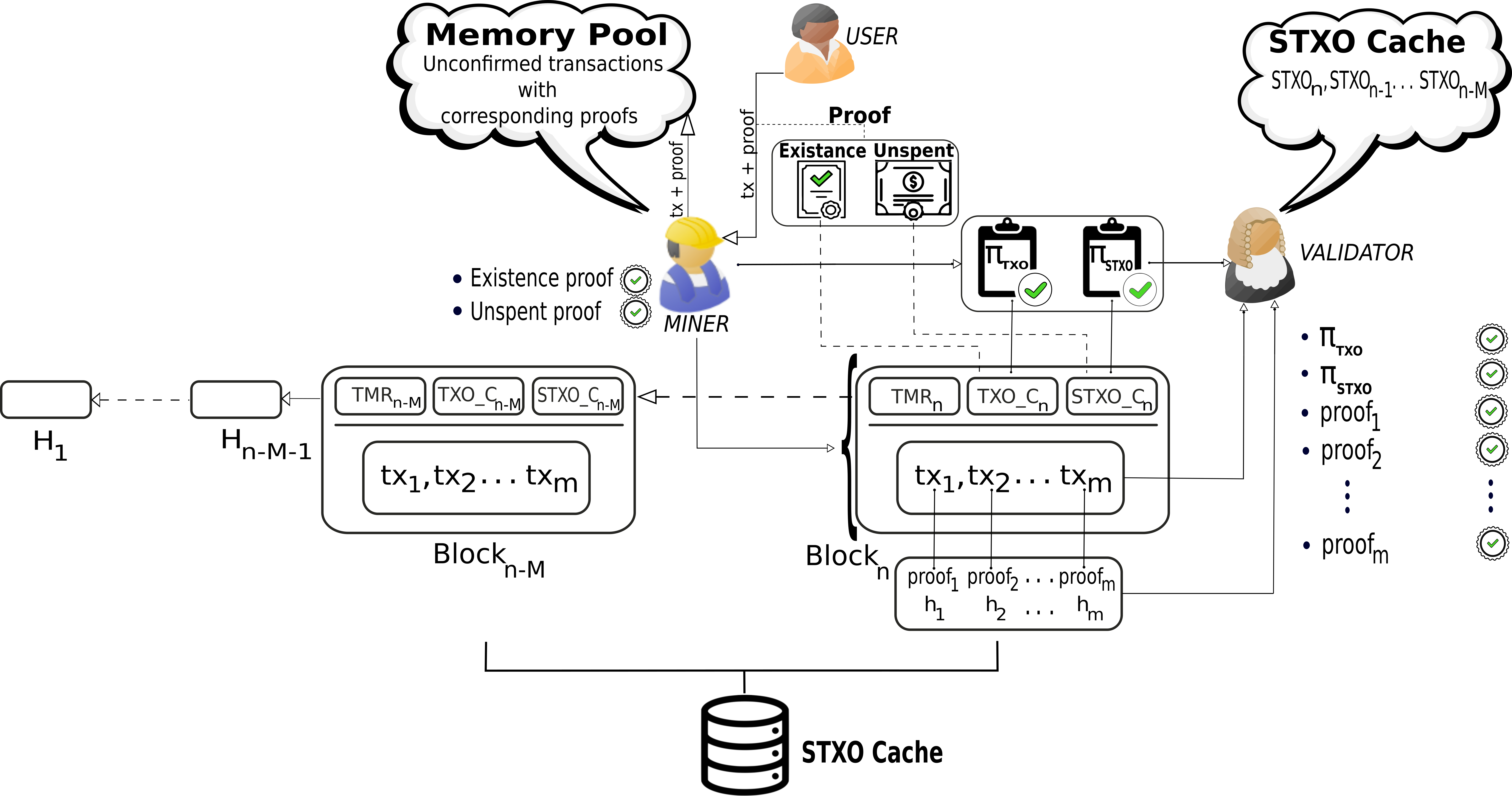}
   \caption{Block validation procedure}
\label{fig:arch_b}
\end{figure*}
We propose a stateless blockchain that is closely built on Boneh's work \cite{batching} and Minichain \cite{minichain}.  The parameters used in our framework are listed in Table \ref{table:symbols}.
We stick to Minichain's model with two commitments for TXO set and STXO set and address the following limitations of the Minichain protocol. 

First, the size of a transaction witness consisting of existence and unspent proofs  for a coin being spent depends on the length of the blockchain $L$ and the maximum number of transactions $m$ in the block with a complexity of $\mathcal{O}(log_2 (m)) + \mathcal{O}(log_2 (L)) + \mathcal{O}(1)$. The length of the chain $L$ is ever-growing, and if a coin is not spent for a sufficiently long duration, then the size of the witness is very large. Second, the information being propagated in the network consists of a block and witnesses for transactions associated with that block. Since the witness size depends on $L$, it affects the end-to-end propagation latency of a block and limits the TPS.

In this work, we propose the RSA accumulator for TXO commitment to address the above limitations. The architecture of the proposed stateless blockchain is shown in Fig.\ref{fig:arch_a}, \ref{fig:arch_b}. \\
\textbf{TXO Commitment:} To reduce the size of the existence proof of a coin in the Minichain protocol and further to improve the efficiency of the network communication, we propose replacing MMR-based TXO commitment with an RSA accumulator. We use batch addition to update the TXO commitment in every block by adding a set of new outputs in $T_n$ of the block $B_n$ to the TXO commitment in the previous block.
When a user wants to spend a coin, he provides a membership witness that specifies the coin included in the TXO commitment. The witness is a single RSA group element, which can be verified using single modular exponentiation.\\
\textbf{STXO Commitment:} The STXO commitment is also an RSA accumulator to represent the STXO set in the block header. The user provides a non-membership witness against the latest STXO commitment to prove that the coin is not yet spent. When a new block is created in the network, the user who did not spend the coin updates the non-membership witness using spent coins in the block.\\
\textbf{STXO cache and Witness height:} When the commitments of the blockchain changes before processing a transaction, the user needs to update the membership and non-membership proofs of the coin being spent, which is very unfriendly to the users. We use the STXO cache introduced by Minichain, which contains the STXOs of the latest $M$ blocks. The validator can still verify the non-membership proof is valid by simply querying the STXO cache without updating the proof to the latest STXO commitment by the user. 

We introduce the Witness height ($h$) to specify the height of the block till where the membership and non-membership proofs got updated so that the user need not update and resubmit the proofs with every new block creation (or change of the state) after submission of the transaction. While submitting proofs for a coin being spent user provides the witness height $h$. The validator can verify the membership and non-membership proofs against the TXO and STXO commitments of the block at a height equal to $h$ such that $(n-M) \leq h \leq n$. Where $n$ is the latest block height, and $M$ represents the length of the STXO cache.
\section{CompactChain Design}
In this section, we discuss the design of the CompactChain protocol - TXO commitment, STXO commitment, witness height, STXO cache, update commitments, transaction witness generation and verification and, finally, update witness. 
Let $\lambda$ be a security parameter and $TXO\_C_n$ and $STXO\_C_n$ commitments to TXO and STXO sets in $n^{th}$ block. The tuple $(w_n(x), u_n(x))$ denotes the membership and non-membership proofs of a coin $x$ in TXO and STXO sets computed/updated till the block at height $n$. For brevity, we have omitted mod $N$ from every modular exponentiation operation in the RSA group with modulus $N$. 
\subsection{Update TXO and STXO Commitments}
The block header in the proposed protocol contains two RSA accumulators called $TXO\_C$ and $STXO\_C$  to accumulate TXO and STXO sets. The RSA accumulator is cumulative. Thereby, the miner updates the commitments for every block using the commitments from the previous block. The functions to update TXO and STXO commitments are described in Algorithm \ref{alg:comm_update}.

The commitments update algorithm takes \textit{inputs}, and \textit{outputs} of a block are input parameters.  It returns updated commitments and proof of correctness for the updated commitments using the RSA accumulator's batching techniques.
For instance, $TXO\_C_{n-1}$ and $STXO\_C_{n-1}$ are the latest commitments to TXO and STXO sets. While generating a new block $B_{n}$, the miner computes the new accumulators $TXO\_C_{n}$ and $STXO\_C_{n}$. Firstly, the $UpdateTXO\_C$ function takes the TXOs (\textit{outputs} of all transactions) from $T_{n}$ and calculate the prime representatives of each \textit{output} in TXOs using $H_{prime}$ function and computes the product of all prime representatives denoted as $p$. Finally, the modular exponentiation of base $TXO\_C_{n-1}$ with exponent $p$ updates the TXO commitment to $TXO\_C_{n}$. 
\begin{algorithm}[!t]
    \caption{Commitments Update Algorithm}
        \label{alg:comm_update}
    $\#$System Initialization
    \begin{algorithmic}[1]
        \Procedure{$Setup$}{$\lambda$}
           \State $\mathbb{G} \gets GGEN(\lambda)$
           \State $ g \stackrel{\$}{\gets}
             \mathbb{G}$
           \State  $TXO\_C_0 \gets 
            g$
            \State $STXO\_C_0 \gets g$
            \State \textbf{return} $TXO\_C_0$, 
              $STXO\_C_0$
        \EndProcedure
    \end{algorithmic}
    \vspace{0.25cm}
    $\#$Update $TXO\_C$ function
    \begin{algorithmic}[1]
        \Procedure{$UpdateTXO\_C$}{$TXO\_C_{n-1}$,$T_{n}$}
           \State $p = 1$
           \For{output in $T_{n}.Outputs$}
           	 \State $p$ \hspace{0.05cm} $*=$ \hspace{0.05cm} $H_{prime}(output)$
           \EndFor
           \State $TXO\_C_{n} \gets \left(TXO
            \_C_{n-1}\right)^{p}$
          
           \State $\Pi_{TXO}$ $\gets$ 
           \textit{ProveNI-PoE}$
           	(TXO\_C_{n-1}, p,$ \\ \hspace{3cm} $TXO\_C_{n})$
           \State \textbf{return} $TXO\_C_{n}$, 
           		$\Pi_{TXO}$
        \EndProcedure
    \end{algorithmic}
\vspace{0.25cm}
    $\#$Update $STXO\_C$ function
    \begin{algorithmic}[1]
        \Procedure{$UpdateSTXO\_C$}{$STXO\_C_{n-1}$, 	
        			$T_{n}$}
           \State $p = 1$
           \For{input in $T_{n}.Inputs$}
           	 \State $p$ \hspace{0.05cm} $*=$ \hspace{0.05cm} $H_{prime}(input)$
           \EndFor
           \State $STXO\_C_{n} \gets \left(STXO
           	\_C_{n-1}\right)^{p}$
          
           \State $\Pi_{STXO} \gets 
           	$\textit{ProveNI-PoE}$
              \left(STXO\_C_{n-1}, p, STXO\_C_{n}\right)$
           \State \textbf{return} $STXO\_C_{n}
            $, $\Pi_{STXO}$
        \EndProcedure
    \end{algorithmic}
     \vspace{0.25cm}
    $\#$Verify updated commitments
    \begin{algorithmic}[1]
        \Procedure{$VerifyCom$}{$T_{n}, TXO\_C_{n-1}, STXO\_C_{n-1},$ $TXO\_C_{n},$ $STXO\_C_{n},  \Pi_{TXO}, \Pi_{STXO}$}
           \State $p_1 = 1$, $p_2 = 1$ 
           \For{output, input in $T_{n}$}
           	 \State $p_1$ \hspace{0.05cm} $*=$ \hspace{0.05cm} $ H_{prime}(output)$

           	 \State $p_2$ \hspace{0.05cm} $*=$ \hspace{0.05cm} $H_{prime}(input)$
           \EndFor
          \State $b_1 \gets $\textit{VerifyNI-PoE}$
          \left(TXO\_C_{n-1}, p_1, TXO\_C_{n
           	}, \Pi_{TXO}\right)$
          \State $b_2 \gets $\textit{VerifyNI-PoE}$ 
          (STXO\_C_{n-1}, p_2, STXO
           	\_C_{n},$ \\ \hspace{1.2cm} $\Pi_{STXO})$
          
           \State \textbf{return} $b_1 \wedge b_2$
        \EndProcedure
    \end{algorithmic} 
 \end{algorithm}
Similarly, the $UpdateSTXO\_C$ function takes the STXOs (\textit{inputs}) from $T_{n}$ as input and update the STXO commitment to $STXO\_C_{n}$.

The miner also generates  NI-PoE proofs $\Pi_{TXO}$ and $\Pi_{STXO}$ for both accumulators using the \textit{ProveNI-PoE}  function to show that the commitments are updated correctly.
The validators need not compute the updated commitments, instead, they can verify the proofs generated by miner using 
\textit{VerifyComm} function. The products calculated in the \textit{VerifyComm} function are much larger, and recalculating the updated commitments is inefficient to the validator. So, the \textit{VerifyNI-PoE} function used in the \textit{VerifyComm} improves the verification efficiency of the validator.
\subsection{Transaction witness Generation $\&$ Verification}
The TXO commitment combines all the generated coins (\textit{outputs}) and, the STXO commitment combines all the spent coins (\textit{inputs}).  
While spending a coin, the user submits a proof along with the transaction to show the validity of the coin. The transaction proof consisting of two parts - First, the proof for the coin's existence in TXO set and the unspent proof for the non-existence of the coin in the STXO set.

The existence proof indicates the coin is generated some time before and is a part of the TXO set represented by  TXO commitment. Similarly, the unspent proof indicates the coin is not yet spent such that it is not a member of the STXO set represented by STXO commitment. In other words, the user needs to generate  membership proof corresponding to the latest TXO commitment and  non-membership proof against the latest STXO commitment. We stick to a procedure that a user generates transaction witness using the TXOs and STXOs of the block where the coin is generated and apply subsequent witness update for every new block of transactions.
\begin{algorithm*}[!t]
    \caption{Witness Generation and Verification Algorithm}
        \label{alg:mem_nonmem}
    \begin{multicols}{2}
    $\#$Prove membership in $TXO\_C_k$ for $x_k$ generated 
    in $B_k$
    \begin{algorithmic}[1]
        \Procedure{$CreateMemWit$}{$x_k,TXO_k
        		, TXO\_C_{k-1}$}
        	\For{$txo$ in $TXO_k$}
        	 \If{$txo \neq x_k$}
           	 \State $p *= H_{prime}(txo)$
           	 \EndIf
           \EndFor
           \State $w_k(x_k) \gets \left(TXO
            \_C_{k-1}\right)^p$
           \State \textbf{return} $w_k(x_k)$
        \EndProcedure
    \end{algorithmic}
    \vspace{0.3cm}
    $\#$Verify Membership Witness for $x_k$ 
     generated in $B_k$
    \begin{algorithmic}[1]
        \Procedure{$\textit{VerifyMemWit}$}{$w_k(x_k)$, 
            $x_k$, $TXO\_C_k$}
           \State $t \gets H_{prime}(x_k)$
           \State $b$ $\gets$ $\left( \left(w_k(x_k)
             \right)^t == TXO\_C_k  \right)$
            \State \textbf{return} $b$
        \EndProcedure
    \end{algorithmic}
    $\#$Prove non-membership against $STXO\_C_k$ for $x_k$  generated in $B_k$
    \begin{algorithmic}[1]
        \Procedure{$CreateNonMemWit$}{$x_k, STXO_k, STXO\_C_{k-1}$}
          \For{$stxo$ in $STXO_k$}
           	 \State $p *= H_{prime}(stxo)$
           \EndFor
           \State $a$, $b$ $\gets$ $Bezout(H_{prime}
           (x_k),p)$
           \State $d \gets \left(STXO\_C_{k-1}\right)^a$
           \State \textbf{return} $u_k(x_k) \gets (d,b)$
        \EndProcedure
     \end{algorithmic}
     \vspace{0.3cm}
    $\#$Verify Non-Membership Witness $u_k(x_k)$ 
     for $x_k$ generated in $B_k$  
    \begin{algorithmic}[1]
        \Procedure{$\textit{VerifyNonMemWit}$}
        {$STXO\_C_k, STXO\_C_{k-1},$ $x_k, u_k(x_k)$}
           \State $t \gets H_{prime}(x_k)$ and $d$, $b$ $\gets$ $u_k(x_k)$

           \State \textbf{return} $\left(d^t 
             \left(STXO\_C_k\right)^b == STXO
             \_C_{k-1}  \right)$
        \EndProcedure
    \end{algorithmic}
    \end{multicols}
  \end{algorithm*}
\subsubsection{Existence Proof (Membership witness in TXO set)} 
The existence proof of a coin in the TXO set is a membership witness corresponding to the latest TXO commitment. The user generates the proof using  the \textit{CreateMemWit} function and the validator verify it using the \textit{VerifyMemWit} function as shown in Algorithm \ref{alg:mem_nonmem}.\\
\textit{\textbf{CreateMemWit}} \textbf{function:} Suppose a coin $x_k$ is generated at block $B_k$, then $x_k \in$ \textit{TXO set} and accumulated to $TXO\_C_k$. The user constructs the membership witness $w_k(x_k)$ for $x_k \in$ \textit{TXO set} as follows - Firstly, the function collects all the TXOs from $TXO_{k}$ excluding $x_k$ and computes their prime representatives. Then, it computes the product of all these prime representatives, denoted as $p$. Finally, the modular exponentiation of the base $TXO\_C_{k-1}$ with exponent $p$ generates the witness $w_k(x_k)$.\\
\textit{\textbf{VerifyMemWit}} \textbf{function:} Suppose $w_k(x_k)$ is a membership witness of the coin $x_k \in$ \textit{TXO set}. Any validator verifies the membership witness by checking $\left(w_k(x_k)\right)^t \stackrel{?}{=} TXO\_C_k$, where $t = H_{prime}(x_k)$.
\subsubsection{Unspent proof (Non-membership witness in STXO set)}
The unspent proof of a coin is a non-membership witness corresponding to the latest STXO commitment. The functions \textit{CreateNonMemWit} and \textit{NonMemWitVerify} are used to generate and verify proof for a new coin,  as shown in Algorithm \ref{alg:mem_nonmem}.\\
\textit{\textbf{CreateNonMemWit}} \textbf{function:} Suppose a coin $x_k$ is generated at block $B_k$ and $x_k \not \in STXOset$, the user constructs a non-membership witness $u_k(x_k)$ of $x_k$ corresponding to $STXO\_C_k$  as follows. The function collects all the STXOs from $STXO_{k}$ and computes their prime representatives. Then, it calculates the product of all these prime representatives denoted as $p$. Since $gcd(t, p) = 1$, where $t = H_{prime}(x_k)$ the unspent proof $u_k(x_k) = \left(d = \left(STXO\_C_{k-1}\right)^a, b\right)$. Where $a$ and $b$ are Bezout's coefficients of $t$ and $p$ such that $at + bp = 1$.
\textit{\textbf{VerifyNonMemWit}} \textbf{function:} Suppose $u_k(x_k)$ is a non-membership witness of $x_k \not \in$ \textit{STXO set} computed by a user with the STXO set till block $B_k$. Any validator can verify it by checking $d^{t}\left(STXO\_C_k\right)^b \stackrel{?}{=} STXO\_C_{k-1}$ as 
\begin{align*}
d^{t}\left(STXO\_C_k\right)^b &= \left(STXO\_C_{k-1}\right)^{at}\left(STXO\_C_{k-1}\right)^{bp}\\  &= STXO\_C_{k-1} 
\end{align*}
\subsection{Transaction Witness Update}
When a user wants to spend a coin, he must submit existence and unspent proofs to show the coin's validity. The submitted proofs might expire due to a delay in network communication, or the miners did not process the transaction in the latest created blocks. In order to bound proofs with the latest commitments, the user must update and resubmit the transaction and witness until the transaction is processed. To address this problem, we require two essential tools. First, the user needs to specify the height of the block $h$ called witness height, where the proofs and commitments are bounded together. The validator checks the membership and non-membership proofs corresponding to $TXO\_C_h$ and $STXO\_C_h$. Second, each validator needs to store the STXOs of the latest $M$ blocks (STXO cache)  in their database to process the transaction in the subsequent $M$ blocks. If the latest block height is greater than $h$, the validator checks that the transaction inputs are consumed in the future blocks by simply querying the STXO cache to avoid the double spending of the coin.
\begin{algorithm}[!t]
    \caption{Witness Update Algorithm}
        \label{alg:wit_update}
    $\#$Update membership proof from block $k$ to $n$
    \begin{algorithmic}[1]
        \Procedure{$UpMemWit$}{$x_k$, $w_k(x_k)$, 
            $TXO_{k+1:n}$}
        	\For{$txo$ in $TXO_{k+1:n}$}
        	 
           	 \State $p *= H_{prime}(txo)$
           	
           \EndFor
          
           \State \textbf{return} $w_n(x_k) \gets 
            \left(w_k(x_k)\right)^p$
        \EndProcedure
    \end{algorithmic}
    \vspace{0.3cm}

$\#$Update non-membership proof from block $k$ to $n$
    \begin{algorithmic}[1]
        \Procedure{$\textit{UpNonMemWit}$}{$x_k, 
          u_k(x_k), STXO_{k+1:n}$}
        	\For{$stxo$ in $STXO_{k+1:n}$}
           	 \State $p *= H_{prime}(stxo)$
           \EndFor
           \State $a_0$, $b_0$ $\gets$ 
                  $Bezout(H_{prime}(x_k),p)$
           \State $r \gets a_0b$, 
             $d^{'} \gets d\left(STXO\_C_{k}
           \right)^r$
           \State \textbf{return} $u_n(x_k) \gets 
            (d^{'},b_0b)$
        \EndProcedure
     \end{algorithmic}
  \end{algorithm}
The most recent witness updated coins can be further updated by the users themselves using the \textit{UpMemWit} and \textit{UpNonMemWit} functions as shown in Algorithm  \ref{alg:wit_update}. The witness update is independent of the size of the sets \textit{TXO set} and \textit{STXO set}, which makes a user no need to store the TXO and STXO sets with them. \\
\textbf{UpMemWit function:} Let $w_n(x_k)$ be the membership witness of the coin $x_k$ generated in block $B_k$. The user can update  till $B_n$ ($n > k$) using \textit{UpdateMemWit}. Firstly, the function gets all the TXOs from block height 
k$+$1 to $n$, denoted as $TXO_{k+1:n}$ and converts them to prime representatives. Then, it computes the product $p$ of all these prime representatives and returns $w_n(x_k)$ using a single modular exponentiation. To verify the updated witness, one check that $\left(w_n(x_k)\right)^{t}$  $\stackrel{?}{=}$ $TXO\_C_n$ (Where, $t = H_{prime}(x_k)$), which holds as 
\begin{align*}
\left(w_n(x_k)\right)^{t}= \left(w_k(x_k)\right)^{pt} = \left(TXO\_C_k\right)^{p} = TXO\_C_n
\end{align*}
\textbf{UpNonMemWit function:} Let $u_k(x_k)$ $=$ $(d,b)$ be non-membership witness of the coin $x_k$ generated in block $B_k$. The user can update  till $B_n$ ($n > k$) using \textit{UpdateNonMemWit}. Firstly, the function gets all the STXOs from block height k$+$1 to $n$ are denoted as  $STXO_{k+1:n}$ and convert them to prime representatives. Then, it calculates the product of all these prime representatives, denoted as $p$. It calculates Bezout coefficients of $t = H_{prime}(x_k)$ and $p$ as $a_0$ and $b_0$. Finally, the updated non-membership witness $u_n(x_k)$ $=$ $\left(d^{'}, b^{'} \right) $. To verify it, one check that, $d^{'t} \left(STXO\_C_m \right)^{b^{'}} \stackrel{?}{=} STXO\_C_{k-1}$, which holds as,
\begin{align*}
d^{'t} \left(STXO\_C_n \right)^{b^{'}} &= d^t \left(STXO\_C_k \right)^{rt} \left(STXO\_C_n \right)^{b_0b} \\
&= d^t \left(STXO\_C_k \right)^{a_0bt} \left(STXO\_C_k \right)^{b_0bp} \\
&= d^t \left(STXO\_C_k \right)^{b\left(a_0t + b_0p\right))} \\ 
&= d^t\left(STXO\_C_k \right)^{b} = STXO\_C_{k-1}
\end{align*}
\subsection{Security Analysis}
In this section we describe the attacker model for double-spend attack of a transaction and analyse the security of the CompactChain agianst this atatck.
\subsubsection{Attacker Model}
We describe the two attacker models on CompactChain construction as follows - 
\begin{enumerate}
    \item \textbf{Double-spend attack:} Let a coin $x$ is generated in block $B_i$ as an output of one of the transactions in $B_i$, i.e., $x \in T_i.Outputs$. $x$ is odd prime. Now, the coin $x \in$ \textit{TXO set} and $x \not \in$ \textit{STXO set}. The owner of the coin $x$ computes the transaction witness $W_x=\left(w_i(x),u_i(x)\right)$ for $x$ from $TXO_i$ and $STXO_i$ using Algorithm \ref{alg:mem_nonmem}. Let the updated witness after the creation of the block $B_n$ is $W^{'}_x=\left(w_n(x),u_n(x)\right)$. When the user wants to spend the coin $x$ in transaction \textit{tx} with $x$ as one of the \textit{inputs} in \textit{tx}, he submits $W^{'}_x$ along with \textit{tx}. Let $n < h < n+M$ such that $tx \in B_h$. Since $x \in STXO_h$, then $x \in $ \textit{STXO set}. i.e., the coin has been spent in block $B_h$. Suppose, a probabilistic polynomial time (PPT) adversary $\mathcal{A}_1$ wants to double-spend the coin $x$ by including it in \textit{$tx'$} (to double-spend the coin which was already spent in  transaction \textit{tx}). Let $k>h$ be the block height at which $\mathcal{A}_1$ computes the witness for \textit{$tx'$}.
The membership proof $w_k(x)$ is easy to prove since $x \in$ \textit{TXO set} and $TXO\_C_k = (w_k(x))^x$. So, the adversary $\mathcal{A}_1$ needs to prove a non-membership witness $u_k(x)$ for coin $x \in $ \textit{STXO set}.
\item For any arbitrary coin $y \notin$ \textit{TXO set}, a PPT adversary $\mathcal{A}_2$ tries to convince the validators of the CompactChain by providing the witness $W_y = \left(w_k(y), u_k(y)\right)$ at a height $k$.  
\end{enumerate}
In both cases, the adversary to succeed, need to prove the membership and non-membership proofs simultaneously in each of the commitments $STXO\_C$ (case $1$) and $TXO\_C$ (case $2$) which is computationally hard as per the undeniability property of the RSA accumulator as per Section~\ref{sec}.
\begin{theorem}
The CompactChain construction is secure under the strong RSA assumption.
\end{theorem}
\begin{proof}
\textbf{Case 1} (double-spend attack): the adversary $\mathcal{A}_1$ returns a witness tuple $W_x=\left(w_k(x),u_k(x)\right)$ for $x \in $ \textit{TXO set} and already spent before height $k$ (i.e., $x$ is already included in the commitment $STXO\_C_{k}$) such that \textit{VerifyMemWit}$ \left(w_k(x),x, TXO\_C_k\right) = 1$ and \textit{VerifyNonMemWit}$ \left(STXO\_C_{k},STXO\_C_{k-1},x,u_k(x)\right) = 1$. So, $\left(d^x . STXO\_C_k\right)^b = STXO\_C_{k-1}$.

We construct an Algorithm  $\mathcal{B}_1$ to break the strong RSA assumption by invoking $\mathcal{A}_1$. $\mathcal{B}_1$ parses $\left(d,b\right) \leftarrow u_k(x)$, $w_k \leftarrow w_k(x)$. $\mathcal{B}_1$ compute $w = \left(STXO\_C_{k-1}\right)^ {\prod_{p \in TXO_k, p \neq x} p}$ and $Z = d.w_k^b$. Then, $Z^x = d^x. w^{xb} = d^x. (STXO\_C_k)^b = STXO\_C_{k-1}$, which contradicts the strong RSA assumption. 

\textbf{Case 2:} The adversary $\mathcal{A}_2$ returns a witness tuple $W_y=\left(w_k(y),u_k(y)\right)$ for $y \notin $ \textit{TXO set} such that \textit{VerifyMemWit}$ \left(w_k(y),y, TXO\_C_k\right) = 1$ and \textit{VerifyNonMemWit}$ \left(STXO\_C_{k},STXO\_C_{k-1},y,u_k(y)\right) = 1$.
So, from these $w_k^y = TXO\_C_k = \left(TXO\_C_{k-1}\right)^{\prod_{x \in TXO_k} x}$.
 
We construct an Algorithm  $\mathcal{B}_2$ to break the strong RSA assumption by invoking $\mathcal{A}_2$. The Algorithm $\mathcal{B}_2$ parses $w_k \leftarrow w_k(y)$ and compute the following - 
\begin{enumerate}
\item Compute $x^{*} = \prod_{x \in TXO_k} x$
\item Compute $a$ and $b$ such that $a.y +b.x^{*} = 1$.
\item Compute $Z = w_k^b . (TXO\_C_{k-1})^a$
\end{enumerate}
Then 
\begin{align*}
Z^y &= w_k^{by}\left(TXO\_C_{k-1}\right)^{ay} \\
&=\left(TXO\_C_{k-1}\right)^{ay + bx^{*}} = TXO\_C_{k-1}
\end{align*}
which contradicts the strong RSA assumption.
\end{proof}
\section{Performance Evaluation}
In this section, we discuss the performance of CompactChain protocol  by comparing it with Boneh's \cite{batching} and Minichain\footnote{$Minichain^{+}$ in \cite{minichain} is considered as Minichain.}  \cite{minichain} protocols. We test the performance of commitments update by miners, transaction witness generation, updation by users. We also test the verification of the NI-PoE proofs for commitments, and transaction witnesses by validators. We also examine the proof size and test performance of the propagation latency of a block in association with transaction throughput. The results shown are averaged over $10$ iterations. 
\begin{table*}[t]
\centering
\caption{Comparison of CompactChain with Boneh and Minichain protocols}\label{table:comparison}
\begin{tabular}{lccc}
\hline
\textbf{ } & \textbf{Boneh} &  \textbf{Minichain} & \textbf{CompactChain} \\
\hline
Accumulator type & RSA(UTXO)  & MMR(TXO) $+$ RSA(STXO)   & RSA(TXO) $+$ RSA(STXO) \\
Commitment Update & $\mathcal{O}(m^2)$  & $\mathcal{O}(m)$ & $\mathcal{O}(m)$ \\ 
Tx Proof generation & $\mathcal{O}(k+d)$ &  $  \mathcal{O}(log(m)) + \mathcal{O}(log(L)) + \mathcal{O}(d)$ & $\mathcal{O}(k+d)$ \\
Tx Proof verification & $\mathcal{O}(1)$ &  $\mathcal{O}(log(m)) + \mathcal{O}(log(L)) + \mathcal{O}(1) $ & $\mathcal{O}(1) $ \\ 
Tx Proof size & $\mathcal{O}(1)$ &  $\mathcal{O}(log(m)) + \mathcal{O}(log(L)) + \mathcal{O}(1) $ & $\mathcal{O}(1) $ \\ 
Tx Proof update & $\mathcal{O}(k+d)$  & $\mathcal{O}(d)$ & $\mathcal{O}(k+d) $ \\ 
\hline 
\end{tabular}
\end{table*}
We have implemented our experiments based on the implementation of \cite{crypto-acc} in C++. We use the RSA modulus of $3072$-bits, $128$-bit prime representative and $256$-bit Merkle root. We run all our experiments for miner, validator, and user tasks on a machine equipped with Intel(R) Core(TM) $i5-8250$U CPU $@$ $1.60$ GHz processor and $8$ GB of RAM. Our experiments are available at \cite{compactChain_exp}.
\subsection{Theoretical Comparisons} 
We discuss the theoretical comparison of CompactChain with Boneh's and Minichain protocols. The comparisons are shown in Table \ref{table:comparison}. All these algorithms depend on the number of coins generated $k$ (\textit{outputs}), the number of coins consumed $d$ (\textit{inputs}), and the number of transactions $m$ in the block.  

\subsection{Commitments update and verification}
In this section, we test the performance of Commitment update and verification discussed in Algorithm \ref{alg:comm_update} and compare with the accumulator updates in Boneh's and Minichain protocols. While creating a new block, the miner must update the commitments due to new TXOs and STXOs.

Fig. \ref{fig:comm_update} shows the comparison of the time consumed for updating the commitments in Boneh's work, Minichain and CompactChain concerning the number of transactions $m$ (we are assuming $k = d$).  Boneh's accumulator update consumes much larger time than Minichain and CompactChain due to the $\mathcal{O}(m^2)$ complexity of the $batchDel$ operation used for accumulator update. For $1000$ transactions, Boneh's accumulator update consumes $\approx 240$ seconds, whereas Minichain and CompactChain took $\approx 1$ second. The Minichain requires one $batchAdd$ operation for $d$ number of spent coins for updating $STXO\_C$. Minichain also updates $TXO\_C$ by adding a new $TMR$ in MMR tree.

In the CompactChain,  two $batchAdd$ operations are required, one each for $STXO\_C$ ($d$ number of spent coins) and $TXO\_C$ ($k$ number of generated coins) updates. The $batchAdd$ depends on the generation of prime representatives for each element to be accumulated. 
While creating a new block, $d$  new elements add to $STXO\_C$, and $k$ elements add to $TXO\_C$.  

To improve the efficiency of our implementation, we performed computation of prime representatives in parallel as each element is independent of others. Then the computation of updating the commitments $TXO\_C$ and $STXO\_C$ are performed in parallel as the set of elements to be added to each commitment are independent. We also exploited parallelism in implementing the Minichain. The slight bias in time taken for commitments update between  CompactChain and Minichain is due to additional computation time for calculating prime representatives of $k$ elements added to $TXO\_C$. In both Minichain and CompactChain, the time to update the commitments increases linearly with the transaction count.
\begin{figure}[t]
  \includegraphics[width=\linewidth, height=0.68\linewidth]{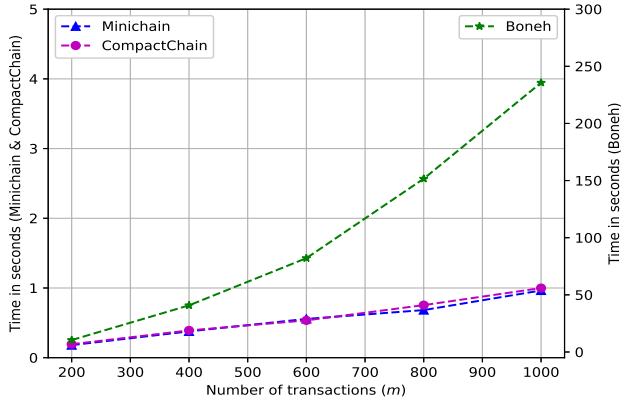}
   \caption{Performance of the commitments update }
\label{fig:comm_update}
\end{figure}
Fig. \ref{fig:comm_ver} shows the time taken by a validator for verifying the commitment updates in Boneh's work, Minichain and CompactChain. The verification depends on the  computation of the prime representatives and the product of the prime representatives. In Minichain, the validator needs to verify a NI-PoE proof for one RSA accumulator, and Merkle proof for MMR peaks update. Whereas  Boneh's work and CompactChain validator verify two NI-PoE proofs.

The validator verifies each NI-PoE proof independently with a constant number of group operations. We have exploited parallelism in commitments verification similar to commitments update. The verification time increases with an increase in transaction count shown Fig. \ref{fig:comm_ver}.
\begin{figure}[t]
  \includegraphics[width=\linewidth, height=0.68\linewidth]{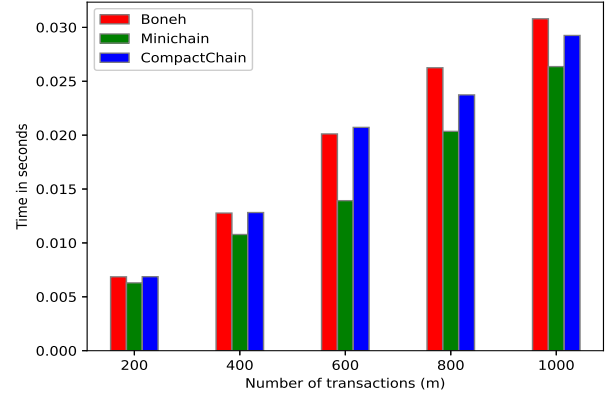}
   \caption{Performance of the commitments verification}
\label{fig:comm_ver}
\end{figure}

\subsection{Transaction Witness generation and verification}
The validator needs to verify the existence and unspent proofs of a transaction in both Minichain and CompactChain. The unspent proof is the same in both, but the existence proof is different in size and structure. In Minichain, the validator needs to verify two Merkle proofs - transaction inclusion proof in TMR, and TMR inclusion proof in the latest TXO commitment. We assume the length of the chain $L = 2^{20}$ and $m = 2^{10}$. In CompactChain, the existence proof is an RSA group element, and the validator verifies it by single modular exponentiation. In Boneh's protocol, the validator verifies a single unspent proof (membership in the UTXO commitment) of the transaction's input in the latest commitment.

Table \ref{table:tx_proof} shows the comparison of the proof size and verification time. The results show that the proof size of the CompactChain has decreased compared with Minichain for the $m$ and $L$ values as mentioned earlier. Since the existence and unspent proofs are independent, the proof verification is performed in parallel. 
\begin{table*}[t]
\begin{center}
\caption{Comparison of Transaction witness size (in bytes) [verification time (in seconds)]}\label{table:tx_proof}
\begin{tabular}{lcccc}
\hline%
& Existence proof & Unspent proof & Verification time per Tx & Verification time for 1000 Txs \\
& (per Tx) & (per Tx) & (parallel execution) & (Parallel execution with 16 threads) \\
\hline
Boneh & -- & $384$ [0.00067]&  0.00067 s & 0.193 s \\
Minichain & $960$ [0.000025] & $400$ [0.0011] & 0.00117 s & 0.306 s\\
CompactChain & $384$ [0.00056] & $400$ [0.0010] & 0.00115 s & 0.303 s\\
\hline
\end{tabular}
\end{center}
\end{table*}

\subsection{Transaction witness update}
In CompactChain, the existence and unspent proofs change with commitments as new blocks are created in the network. We assume that the number of \textit{inputs} and \textit{outputs} are equal and, Fig. \ref{fig:witness_update} shows that the transaction witness update time for Boneh's, Minichain and CompactChain. In Minichain, the time taken for existence proof update is negligible. Boneh's witness update requires two sequential events for batch deletion of \textit{inputs} from the witness, then batch addition of \textit{outputs} to the witness.
In CompactChain, we implemented a witness update similar to the commitments update by exploiting the parallelism.
The time taken by a user to update the witness increases linearly with the transaction count in all three protocols.
\begin{figure}[t]
  \includegraphics[width=\linewidth, height=0.68\linewidth]{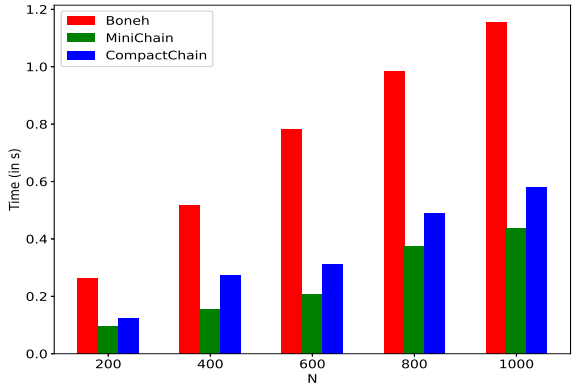}
   \caption{Performance of the transaction witness update}
\label{fig:witness_update}
\end{figure}

\subsection{Memory Consumption}
We compare the memory consumption of existing stateless blockchains against Bitcoin. Comparison is analyzed on three aspects (i) RAM Usage, (ii) Disk Usage, and (iii) Memorypool Usage. The comparison was performed on real-time data of Bitcoin from Jan $2017$ to Jan $2022$. 
\begin{figure}[t]
  \includegraphics[width=0.9\linewidth, height=0.6\linewidth]{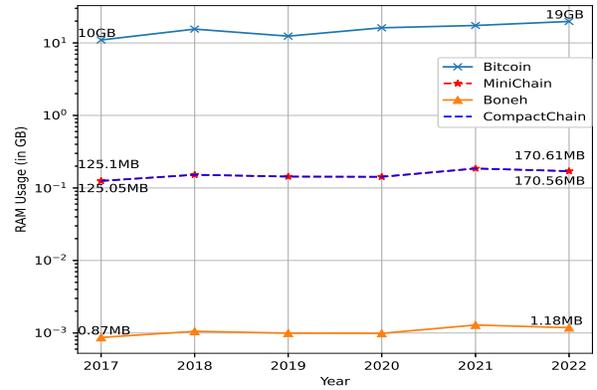}
   \caption{Comparison of RAM Usage}
\label{fig:ram}
\end{figure}
RAM Usage, considers memory usage while storing stxo-cache in the case of CompactChain and MiniChain, and wheras Bitcoin requires the whole UTXO-set. Since Boneh has no cache mechanism we consider only current block. Fig. \ref{fig:ram} shows the RAM usage from $2017$ to $2022$. Bitcoin consumes the highest memory, 19GB due to the ever-growing UTXO set size. CompactChain consumes $170.61$ MB of RAM. MiniChain consumes $0.05$ MB lesser than CompactChain because CompactChain's TXO commitment is $356$ bytes larger than MiniChain's. Boneh consumes the least space due to the lack of caching mechanism, however, underpinned by commitment update speeds. Stateless blockchains can be practically realizable on resource constraint IoT devices as validators due to their low RAM usage.
\begin{figure}[t]
  \includegraphics[width=0.9\linewidth, height=0.6\linewidth]{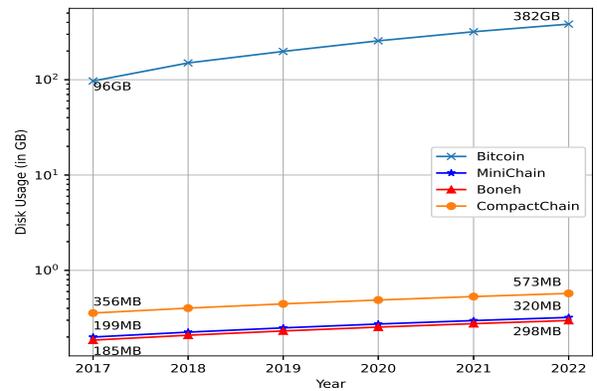}
   \caption{Comparison of Disk Usage}
\label{fig:disk}
\end{figure}
\begin{figure}[t]
\centering
  \includegraphics[width=0.9\linewidth, height=0.6\linewidth]{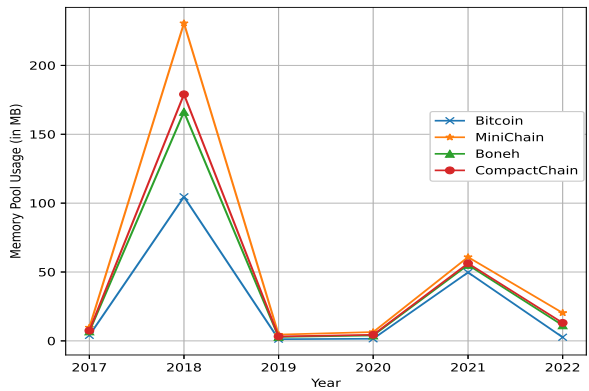}
   \caption{Comparison of Memorypool Usage}
\label{fig:mempool}
\end{figure}
Disk Usage, considers the size required to store the whole blockchain. Fig. \ref{fig:disk} shows the disk usage, as expected Bitcoin consumes the highest because one has to store entire blocks of the blockchain. In contrast, stateless blockchains need to store only block headers, therefore consuming very less space. As pointed out earlier, CompactChain consumes slightly higher header space than MiniChain. Nevertheless, CompactChain consumes $573$ MB only which enables edge devices to act as Validators.
\begin{table*}[ht]
\centering
\caption{Comparison of Maximum TPS}
\begin{center}
\begin{tabular}{lccc}
\hline%
& Boneh & Minichain & CompactChain \\
\hline
Tx Verification latency (in sec) (for 500 Txs) [Max TPS] & 0.193 [2590.67]  & 0.306 [1634] & 0.303 [1650] \\
Commitments Update latency (in sec) [Max TPS] & 235.62 [2.12]& 0.97 [515.46] & 0.99 [505] \\
Consensus latency (in sec) [Max TPS] & 2.08 [240.38] & 3.57 [140] & 3.03 [165]\\
Maximum TPS  & 2.12 & 140 & 165 \\
\hline
\end{tabular}
\label{table:TPS}
\end{center}
\end{table*}

Memory Pool or Mempool\footnote{Mempool is a database which consists of the valid unconfirmed transactions in the network. A high Mempool size indicates more network traffic which will result in longer average confirmation time and higher priority fees. While creating a new block miner or block producer pick the transactions from Mempool.} Usage considers the size of unconfirmed transactions received and corresponding proofs (if any). A validator in the stateless blockchain need to store the proofs for transactions of the Mempool, whereas Bitcoin does not. Fig. \ref{fig:mempool} indicates that Bitcoin least space, however, MiniChain consumes enormously high space due to large MMR proofs compared to constant sized Non-Membership witness of CompactChain. Boneh's implementation has only single proof thereby consuming least space among stateless-blockchains.
\subsection{Propagation latency}
We examine the propagation latency of a block as per the information propagation in the bitcoin network as described in \cite{info}. We have conducted  event-driven simulations \cite{optimal, UL-blockDAG, scalable} for propagation latency. We choose a total of $13000$ \cite{bitnodes} nodes with $10$ miners having hash rate distribution \cite{bitcoin_charts} as per in the bitcoin network. We choose the average upload bandwidth of the nodes as $50$ Mbps \cite{speedtest}.

We assume the average size of a transaction (with single input and single output) as $250$ bytes with a total block size of $0.25$ MB ($1000$ transactions). We use the total block validation time of a validator from Fig. \ref{fig:comm_ver} (commitments verification time) and Table \ref{table:tx_proof} (transactions verification time). The information to be propagated  in the network also consists of the transaction proof with sizes as listed in Table \ref{table:tx_proof}. Fig. \ref{fig:prop_delay} shows the total information propagation from the first observation of a block at the miner to reach all nodes in the network.
\begin{figure}[t]
  \includegraphics[width=\linewidth, height=0.7\linewidth]{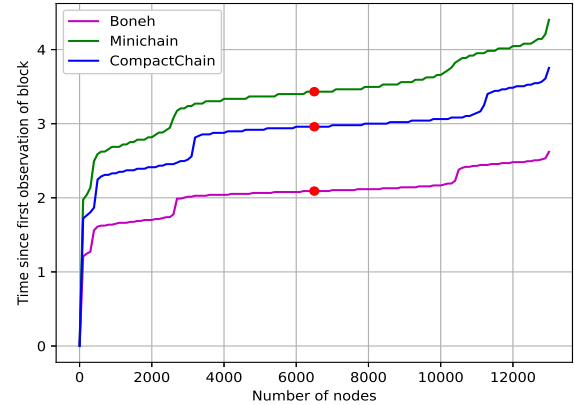}
   \caption{Performance of propagation latency of a block in the network}
\label{fig:prop_delay}
\end{figure}

In Bitcoin \cite{bitcoin}, the consensus is defined in probabilistic terms and a block is valid if at least $50$ percent of the nodes receive it. So, \textit{consensus latency} is defined as time to reach a block to $50$ percent of the nodes in the network. In Fig. \ref{fig:prop_delay}, the dots in red colour indicate the consensus latency in each protocol. 
\subsection{Maximum TPS}
We test the maximum TPS of the system by considering the performances of transaction verification, commitment updates, and consensus latency. We assume each block with $500$ transactions, and each transaction consists of two \textit{inputs} and two \textit{outputs}. Table \ref{table:TPS} shows the comparison of maximum TPS in all  three frameworks. The results show that the maximum TPS in Boneh's work is restricted to $2.12$ TPS due to the complex commitment update.  Although having two RSA accumulators CompactChain has improved maximum TPS over Minichain due to a reduction in the size of the information to be propagated in the network.

\section{Conclusions and Future Research}
\textit{(i) Summary of research and conclusion} - In this paper, we propose a stateless UTXO-model blockchain called compactChain, where the state of the blockchain is comprised of  two commitments similar to Minichain. We introduce the RSA accumulator-based TXO commitment, in contrast to the MMR-based TXO commitment in Minichain. We analyze the security of the system based on the strong RSA assumption.

\textit{(ii) Comparison with Boneh's work and Minichain}- the theoretical and  experimental results show the improvement in the time complexity of the commitment update from $\mathcal{O}(m^2)$ to $\mathcal{O}(m)$ compare with Boneh's stateless blockchain. The CompactChain reduce the transaction proof size from $\mathcal{O}(log(m)) + \mathcal{O}(log(L))+  \mathcal{O}(1)$ to $\mathcal{O}(1)$ compare with Minichain without compromising the system throughput.

\textit{(iii) Practical implications} - CompactChain allows resource constraint devices to become validators due to marginal RAM usage, disk space and less memory pool size, retaining the decentralization property of the blockchain.

\textit{(iv) Paper limitations and future work} -
The future work will focus on an alternate consensus algorithm to the PoW, which is based on the proof of commitments (PoC) by using the number of transactions as a difficulty level. Sharding technology divides a whole Blockchain network into multiple groups and allows participating nodes to process and store the disjoint set of transactions. For securing the sharded-blockchains, the
participating nodes should be reshuffled randomly among the shards, making a node download the state of the newly allocated shard. We focus on integrating the CompactChain with sharded technology to avoid downloading a state for every reshuffle that improves startup time.

\section*{Acknowledgment}
This work is supported by the project \textbf{Indigenous 5G Test Bed (Building an end to end 5G Test Bed) in India}, Dept. of Telecommunication Networks $\&$ Technologies (NT) Cell, Government of India.

\begin{IEEEbiography}
[{\includegraphics[width=1in,height=1.25in,clip,keepaspectratio]{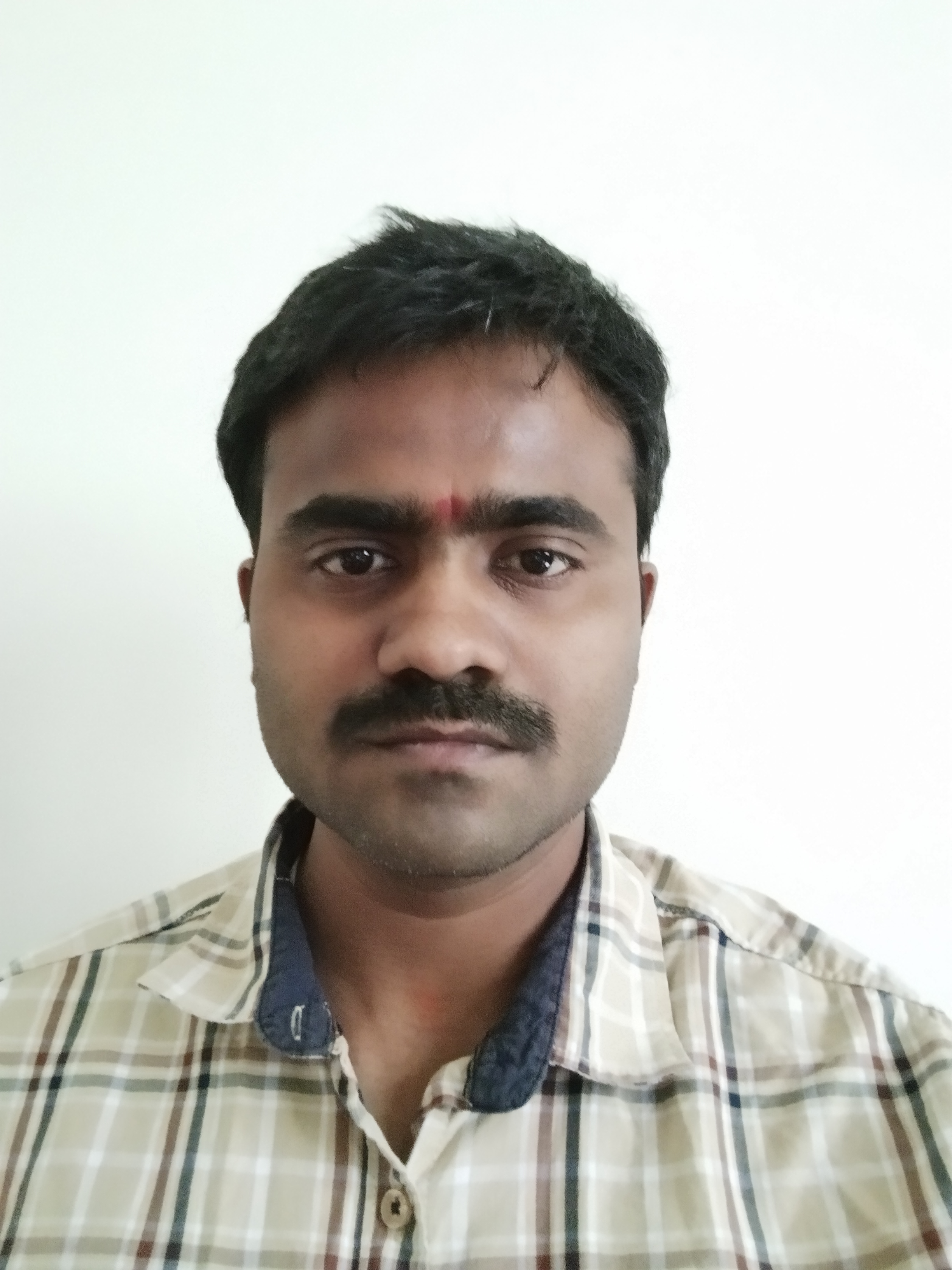}}]{\textbf{B Swaroopa reddy} is currently pursuing the Ph.D degree in the Department of Electrical Engineering, Indian Institute of Technology Hyderabad, India from 2017. He received the B.Tech degree in Electronics and Communication Engineering from Sri Krishna Devaraya University, Anantapur, India, in 2009. From 2010 to 2017, he was with the Bharath Sanchar Nigam Limited (BSNL), a Public Sector Unit under the Government of India. 
His research interests include scalability of Blockchain networks, Privacy and Security in Blockchain, ZK-Rollups, Decentralized Identities (DiD) and Verifiable Credentials (VC).
}
\end{IEEEbiography}

\begin{IEEEbiography}
[{\includegraphics[width=1in,height=1.25in,clip,keepaspectratio]{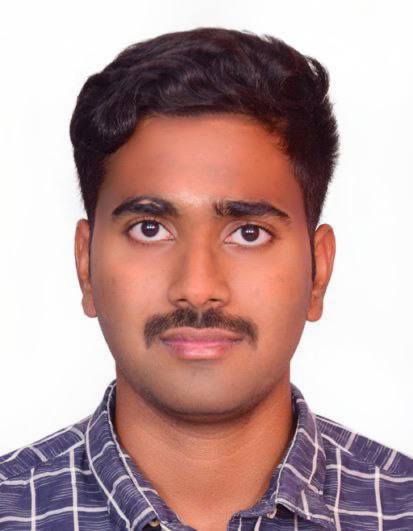}}]{\textbf{T Uday Kiran Reddy} is currently pursuing the B.Tech degree in the Department of Electrical Engineering, Indian Institute of Technology Hyderabad, India from 2019. His current research interests include the areas of Stochastic Optimisation, Signal Processing, Machine Learning, and Blockchain.
}
\end{IEEEbiography}
\end{document}